\documentclass{article}

\usepackage{amssymb,amsmath,amsthm,tikz}
\usetikzlibrary{calc}
\usepgflibrary{arrows}
\usepackage{caption,subcaption}

\newtheorem{lemma}{Lemma}
\newtheorem{theorem}{Theorem}

\theoremstyle{definition}
\newtheorem{definition}{Definition}
\theoremstyle{remark}

\newcommand{\AAUL}{\updownarrow}

\newcommand{\propd}{\mathit{propd}}
\newcommand{\refa}{\mathit{ref}_a}
\newcommand{\return}{\mathit{return}}
\newcommand{\M}{\mathcal{M}}

\newcommand{\up}{\mathit{up}}
\newcommand{\down}{\mathit{down}}
\newcommand{\leftt}{\mathit{left}}
\newcommand{\rightt}{\mathit{right}}

\newcommand{\north}{\mathit{north}}
\newcommand{\south}{\mathit{south}}
\newcommand{\east}{\mathit{east}}
\newcommand{\west}{\mathit{west}}

\renewcommand{\phi}{\varphi}

\title{The Undecidability of Arbitrary Arrow Update Logic}
\author{Hans van Ditmarsch \and Wiebe van der Hoek \and Louwe B. Kuijer}
\date{}

\begin{document}
\maketitle

\begin{abstract}
Arbitrary Arrow Update Logic is a dynamic modal logic that uses an arbitrary arrow update modality to quantify over all arrow updates. Some properties of this logic have already been established, but until now it remained an open question whether the logic's satisfiability problem is decidable. Here, we show that the satisfiability problem of Arbitrary Arrow Update Logic is co-RE hard, and therefore undecidable, by a reduction of the tiling problem.

\end{abstract}
\section{Introduction}
In the field of Dynamic Epistemic Logic, various kinds of \emph{updates} are used to model events that change the information state of agents. These kinds of updates include public announcements \cite{plaza_1989, baltag_1998}, action models \cite{baltag_1998} and arrow updates \cite{AUL}, among others. Somewhat recently, there has been a trend to enrich logics with ``arbitrary'' versions of such updates. See for example Arbitrary Public Announcement Logic (APAL) \cite{APAL}, Group Announcement Logic (GAL) \cite{Agotnes2010}, Arbitrary Action Model Logic (AAML) \cite{hales_2013}, and Arbitrary Arrow Update Logic (AAUL) \cite{AAUL_AiML,AAUL_AIJ}. The intuition behind these ``arbitrary'' operators is that they represent universal quantification over their non-arbitrary counterpart; for example, $\phi$ is true after an arbitrary public announcement if $\phi$ is true after every public announcement.

For some of these logics, namely APAL and GAL, the satisfiability problem is undecidable \cite{french_2008,agotnes_2014}. The satisfiability problem of AAML, on the other hand, is decidable \cite{hales_2013}. For AAUL, it remained unknown whether the satisfiability problem is decidable. Here, we show that AAUL's satisfiability problem can encode the tiling problem \cite{wang_1961}. Because the tiling problem is known to be co-RE complete \cite{berger_1966}, this shows that the satisfiability problem of AAUL is co-RE hard.

%The satisfiability problem of the logics with ``arbitrary'' operators is typically undecidable.\footnote{AAML is an exception to this rule; the satisfiability problem of AAML is decidable because AAML is no more expressive than basic modal logic.} Notably, APAL was proven to be undecidable in \cite{french_2008} and the related Group Announcement Logic (GAL) in \cite{agotnes_2014}. Here, we show that the satisfiability problem of AAUL is also undecidable. 

%Given that APAL and GAL are undecidable, the undecidability of AAUL is not a great surprise. Furthermore, we show the undecidability of AAUL by a reduction of the same undecidable problem that was used to show the undecidability of APAL and GAL: the tiling problem. So there are some similarities between our proof for the undecidability of AAUL and the existing proofs for APAL and GAL. The reduction itself is quite different from existing ones, however.

The structure of this paper is as follows. First, in Section~\ref{sec:semantics} we introduce the syntax and semantics of AAUL. 
%Then, in Section~\ref{sec:tiling} we provide a brief definition of the tiling problem. 
Then, in Section~\ref{sec:reduction} we provide a brief definition of the tiling problem and show that it can be encoded in the satisfiability problem of AAUL.

\section{AAUL Syntax and Semantics}
\label{sec:semantics}
Let $\mathcal{P}$ be a countable set of propositional variables and $\mathcal{A}$ a finite set of agents. We assume that $|\mathcal{A}|\geq 6$.
\begin{definition}
The language $\mathcal{L}_\mathit{AAUL}$ of AAUL is given by the following normal forms:
\begin{align*}
	\phi ::= {} & p \mid \neg \phi \mid \phi \wedge \phi \mid \square_a\phi \mid [U]\phi \mid [\AAUL]\phi\\
	U::= {}& (\phi,a,\phi)\mid U, (\phi,a,\phi)
\end{align*}
where $p\in \mathcal{P}$ and $a\in \mathcal{A}$. The language $\mathcal{L}_\mathit{AUL}$ is the fragment of $\mathcal{L}_\mathit{AAUL}$ that does not contain $[\AAUL]$.
\end{definition}

We use $\vee, \rightarrow, \leftrightarrow, \lozenge, \langle U\rangle, \langle \AAUL\rangle, \bigvee$ and $\bigwedge$ in the usual way as abbreviations. Furthermore, we slightly abuse notation by identifying the list $U=(\phi_1,a_1,\psi_1),\cdots,\allowbreak (\phi_k,a_k,\psi_k)$ with the set $U=\{(\phi_1,a_1,\psi_1),\cdots,(\phi_k,a_k,\psi_k)\}$.

AAUL is evaluated on standard multi-agent Kripke models.
\begin{definition}
A model $\M$ is a triple $\M=(W,R,V)$ where $W$ is a set of states, $R:\mathcal{A}\rightarrow 2^{W\times W}$ assigns to each agent an accessibility relation and $V:\mathcal{P}\rightarrow 2^W$ is a valuation.
\end{definition}
Note that we are using the class of all Kripke models. This is unlike APAL and GAL, which are typically considered on the class of S5 models. 

Now, let us consider the semantics of AAUL. We start by giving the formal definition, after the definition we briefly discuss the intuition behind some of the operators.
\begin{definition}
Let $\M=(W,R,V)$ be a model and let $w\in W$. The satisfaction relation $\models$ is given by
  \[
  \begin{array}{lcl}
    \M,w\models p &\text{iff}& w\in V(p)
    \\
    \M,w\models\neg\varphi &\text{iff}& \M,w\not\models\varphi
    \\
    \M,w\models(\varphi\wedge\psi) &\text{iff}& 
    \M,w\models\varphi \text{ and } \M,w\models\psi
    \\
    \M,w\models \square_a\varphi &\text{iff}&
    \M,v\models\varphi \text{ for each } v \text{ such that } (w,v)\in R(a) 
    \\
    \M,w\models[U]\varphi & \text{iff} & (\M*U),w\models\varphi\\
    \M,w\models [\AAUL] \varphi &\text{iff}& \M,w \models [U]\varphi
        \text{ for each } U \in L_\mathit{AUL}\\
	\end{array}\]
	where $(\M*U)$ is given by:
	\[
    \begin{array}{lcl}
    \M*U & = & (W,R^U,V)\\
    R^{U}(a) & = &
    \{ (v,v')\in R(a) \mid \exists(\varphi,a,\varphi')\in U:
    \\
    &&
    \phantom{\{ v'\in R^\M_a(v) \mid}
    (\M,v\models\varphi \text{ and } \M,v'\models\varphi')
    \}
  \end{array}
  \renewcommand{\arraystretch}{1.0}
  \]
\end{definition}

A full discussion of the applications of AAUL and of the intuitions behind the semantics of arrow updates and arbitrary arrow updates is outside the scope of this paper. For such a discussion, see \cite{AUL} and \cite{AAUL_AIJ}. However, in order to understand the undecidability proof it is important to understand the semantics of AAUL. We therefore do provide a very brief explanation of the intuition behind and the semantics of AAUL.

Although our goal is to understand AAUL, is is useful to start by considering public announcements. We assume that the reader is familiar with public announcement logic, if not see for example \cite{baltag_1998}. A public announcement $[\psi]$ informs all agents that $\psi$ is true. As a result, every possible world that the agents previously considered possible that does not satisfy $\psi$ is rejected after the announcement, since it is incompatible with the new information. Semantically, this corresponds to a model $\M$ being transformed into a model $\M*\psi$ where all $\neg \psi$ states of $\M$ have been removed.

Like public announcements, arrow updates provide agents with new information. Unlike with public announcements, however, the new information provided by an arrow update can (i) differ per agent and (ii) depend on the current state. A typical example is a card game, where cards have been dealt face down. Now, agent $a$ picks up her hand of cards and looks at it. Obviously, the information that $a$ gains from this action is different than the information the other agents gain: $a$ learns what her cards are whereas the other agents only learn that $a$ now knows what her cards are. It is perhaps less obvious that the information that $a$ gains also differs per state. Suppose that $a$ has been dealt the 7 of Hearts. Then by looking at her cards $a$ learns that she has the 7 of Hearts. If, on the other hand, $a$ has been dealt the 8 of Clubs, then she learns that she has the 8 of Clubs. Learning that you have the 7 of Hearts is different from learning that you have the 8 of Clubs, so the information given to $a$ depends on the state of the world.

With arrow updates we formalize the information that the agents gain in such a situation. In principle, we could do this in two ways: we could specify the things that are \emph{incompatible} with the new information, or the things that are \emph{compatible}. We choose to follow public announcements in this aspect, so just like $[\psi]$ says that the new information is compatible with $\psi$, we use an arrow update $U$ to specify the information that is compatible with $U$. Since the information gained in an arrow update can depend on the agent and on the current state, we need triple $(\phi,a,\psi)$. We call such triple \emph{clauses}; they can be read as ``if the current state satisfies $\phi$, then the information provided to agent $a$ is compatible with $\psi$.''

An arrow update $U$ is a finite set clauses, $U=\{(\phi_1,a_1,\psi_1),\cdots,(\phi_k,a_k,\psi_k)\}$ (where it is possible that $\phi_i=\phi_j$, $a_i=a_j$ or $\psi_i=\psi_j$ for $i\not = j$). This still leaves the decision of what to do if a state matches multiple clauses. Suppose, for example, that $(\phi_1,a,\psi_1), (\phi_2,a,\psi_2)\in U$ and that a state satisfies both $\phi_1$ and $\phi_2$. There are several options for how to interpret this situation, we choose to interpret it disjunctively: if a state satisfies $\phi_1$ and $\phi_2$, then any state that satisfies $\psi_1$ \emph{or} $\psi_2$ is consistent with the new information.

On the semantical level, this means that $\M*U$ should contain exactly those arrows of $\M$ that match at least one clause of $U$, where we say that $(w_1,w_2)\in R(a)$ matches $(\phi_1,a_1,\psi_1)$ if and only if $\M,w_1\models \phi_1$, $a=a_1$ and $\M,w_2\models \psi_1$.

Arbitrary arrow updates then quantify over such arrow updates. However, in order to avoid circularity we restrict this quantification to those arrow updates that do not themselves contain an arbitrary arrow update $[\AAUL]$. So $\M,w\models [\AAUL]\phi$ if and only if $\M,w\models [U]\phi$ for all $\phi\in \mathcal{L}_\mathit{AUL}$.

\section{Reducing the Tiling Problem}
\label{sec:reduction}

\subsection{The Tiling Problem}
\label{sec:tiling}
We will prove the undecidability of AAUL by a reduction of the tiling problem. The tiling problem was introduced in \cite{wang_1961} and can be defined as follows.
\begin{definition}
Let $C$ be a finite set of colors. A \emph{tile type} is a function $i:\{\north,\south,\east,\west\}\rightarrow C$.

An instance of the tiling problem is a finite set $\mathit{types}$ of tile types. A solution to an instance of the tiling problem is a function $\mathit{tiling}:\mathbb{Z}\times\mathbb{Z}\rightarrow \mathit{types}$ such that, for every $(z_1,z_2)\in \mathbb{Z}\times\mathbb{Z}$,
\begin{align*}
\mathit{tiling}(z_1,z_2)(\north) = {} & \mathit{tiling}(z_1,z_2+1)(\south)\\
\mathit{tiling}(z_1,z_2)(\east) = {} & \mathit{tiling}(z_1+1,z_2)(\west).
\end{align*}
\end{definition}
The tiling problem was shown to be undecidable in \cite{berger_1966}. In fact, the tiling problem is co-RE complete. Therefore, by reducing the tiling problem to the satisfiability problem of AAUL, we show that the latter problem is co-RE hard. Whether AAUL's satisfiability problem is co-RE is not currently known.

\subsection{Encoding the Tiling Problem in AAUL}
%We want to encode the tiling problem in AAUL. Since AAUL is a modal logic, the most straightforward way to do so is to represent every point in $\mathbb{Z}\times\mathbb{Z}$ with a state. For every tile type $i$ we then use a propositional variable $p_i$ to represent the state containing a tile of type $i$. 

We want to encode the tiling problem in AAUL. So for every instance $\mathit{types}$ of the tiling problem we define a formula $\psi_\mathit{types}$ of AAUL that is satisfiable if and only if $\mathit{types}$ can tile the plane. The strategy for doing this is as follows.

We represent each point of $\mathbb{Z}\times \mathbb{Z}$ by a state. For every $i\in \mathit{types}$ we then use a propositional variable $p_i$ to represent ``the current state contains a tile of type $i$.'' For every $c\in C$ we use propositional variables $\north_c$ (resp.\ $\south_c, \east_c, \west_c$) to represent the northern (resp.\ southern, eastern, western) edge of the current tile having color $c$. Finally, we use relations $\up, \down, \leftt$ and $\rightt$ to represent one tile being above, below, to the left and to the right, respectively, of the current tile.

In addition to the states $(n,m)$ that correspond to points in $\mathbb{Z}\times\mathbb{Z}$, we also use an auxiliary state $s_0$. This state $s_0$ is not part of the grid, and does not contain any tile. Instead, it is the state where $\psi_\mathit{types}$ will be evaluated. We therefore also refer to $s_0$ as the origin state. In order to distinguish $s_0$ from the states that are part of the grid we use the propositional variable $p$, which holds on $s_0$ but not on any $(n,m)$.

Now, given any state $(n,m)$, it is relatively easy to check whether the constraints of a tiling are satisfied locally. For example, $\bigvee_{i\in \mathit{types}}p_i \wedge \bigwedge_{i\not = j \in \mathit{types}}\neg(p_i\wedge p_j)$ holds if and only if the current state has exactly one type of tile, and $\bigwedge_{c\in C}(\north_c\rightarrow \square_\up\south_c)$ holds if and only if the northern color of the current tile matches the southern color of the tile above.

Making sure that the \emph{global} constraints of a tiling are satisfied is harder, though. We do this in the following way. Firstly, we take a relation $b$, and force it to act as a kind of transitive closure\footnote{The precise properties of $b$ are more complicated than this, but for the purpose of this informal introduction to the proof it suffices to think of it as a transitive closure.} over $\up, \down, \leftt$ and $\rightt$. So while $\bigvee_{i\in \mathit{types}}p_i \wedge \bigwedge_{i\not = j \in \mathit{types}}\neg(p_i\wedge p_j)$ says that the current state has exactly one tile type, the formula $\square_b\bigvee_{i\in \mathit{types}}p_i \wedge \bigwedge_{i\not = j \in \mathit{types}}\neg(p_i\wedge p_j)$ says that \emph{all} states (except the current one\footnote{Recall that $\psi_\mathit{types}$ will be evaluated in $s_0$, so the ``current state'' in question is $s_0$, which does not have a tile type.}) have exactly one tile type. Secondly, we enforce a grid-like structure onto the domain.

With the above in mind, let us define the formula $\psi_\mathit{types}$.

\begin{definition}
Let $\mathit{types}$ be an instance of the tiling problem. The formula $\psi_\mathit{types}$ is given by
\begin{align*}\psi_\mathit{types} := {} & \psi_1 \wedge \psi_2\wedge \bigwedge_{x\in D}(\psi_{3,x}\wedge \psi_{4,x}\wedge \propd_x\wedge \return_x)\\
 & \wedge \mathit{inverse}\wedge \mathit{commute}\\
 & \wedge \mathit{one\_tile} \wedge \mathit{one\_color} \wedge \mathit{tile\_colors} \wedge \mathit{tile\_match}
\end{align*}
where
{
\allowdisplaybreaks
\begin{align*}
D := {} & \{\up,\down,\leftt,\rightt\}\\
\psi_1 := {} & \refa \wedge p \wedge \lozenge_b \top \wedge \square_b \neg p\\
\psi_2 := {} & \square_b (\refa \wedge \lozenge_b p)\wedge [\AAUL](\lozenge_a\top \rightarrow \square_b\square_b\lozenge_a \top)\\
\psi_{3,x} := {} & \square_b (\lozenge_x (\neg p \wedge \refa \wedge \lozenge_b p) \wedge [\AAUL](\lozenge_x\lozenge_a\top\rightarrow \square_x\lozenge_a\top))\\
\psi_{4,x} := {} & [\AAUL](\lozenge_a \top \rightarrow \square_b\square_x\square_b\lozenge_a\top)\\
\refa := {} & \lozenge_a \lozenge_a\top \wedge [\AAUL]\neg \lozenge_a\square_a\bot\\
\propd_x := {} & \square_b [\AAUL] ((\square_a\bot \wedge \lozenge_x \lozenge_a\top \wedge \lozenge_b(\lozenge_b\top\wedge \square_b\lozenge_a\top)\wedge \\
&\langle \AAUL\rangle(\lozenge_x \lozenge_a\top\wedge\lozenge_b\lozenge_b\square_a\bot))\rightarrow [U_\mathit{x}]\langle \AAUL\rangle(\lozenge_x \lozenge_a\top\wedge\lozenge_b\lozenge_b\square_a\bot))\\
U_\mathit{x} := {} & (p\vee \square_a\bot,b,\top),(\top,a,\top),(\square_a\bot,x,\top)\\
\return_x := {} & \square_b \langle \AAUL\rangle (\square_a\bot \wedge \lozenge_b\top \wedge 
%\lozenge_x \langle\AAUL\rangle
(\lozenge_a\top \wedge \lozenge_b (\lozenge_b\top\wedge \square_b\lozenge_a\top)\wedge \\ & [\AAUL](\lozenge_a\top\rightarrow \square_b\square_b\lozenge_a\top)))\\
\mathit{inverse} := {} & \square_b [\AAUL](\square_a\bot \rightarrow(\square_\up\square_\down\square_a\bot\wedge \square_\down\square_\up\square_a\bot \\ & \wedge \square_\leftt\square_\rightt\square_a\bot \wedge \square_\rightt\square_\leftt\square_a\bot))\\
\mathit{commute} := {} & \square_b [\AAUL]\bigwedge_{(x,y)\in E}(\lozenge_x\lozenge_y\square_a\bot \rightarrow \square_y\square_x\square_a\bot)\\
E := {} & \{(\up,\leftt), (\up,\rightt),(\down,\leftt),(\down,\rightt), \\ 
&(\leftt,\up), (\leftt,\down),(\rightt,\up),(\rightt,\down)\}\\
\mathit{one\_tile} := {} & \square_b (\bigvee_{i\in \mathit{tiles}} p_i \wedge \bigwedge_{i\not = j \in \mathit{tiles}}\neg (p_i\wedge p_j))\\
\mathit{one\_color} := {} & \square_b \bigwedge_{c\in C}(\north_c\rightarrow \bigwedge_{d\in C\setminus\{c\}}\neg \north_d) \wedge \\ & \square_b \bigwedge_{c\in C}(\south_c\rightarrow \bigwedge_{d\in C\setminus\{c\}}\neg \south_d)\wedge \\
& \square_b \bigwedge_{c\in C}(\east_c\rightarrow \bigwedge_{d\in C\setminus\{c\}}\neg \east_d)\wedge \\ & \square_b \bigwedge_{c\in C}(\west_c\rightarrow \bigwedge_{d\in C\setminus\{c\}}\neg \west_d)\\
\mathit{tile\_colors} := {} & \square_b \bigwedge_{i\in \mathit{tiles}}(p_i\rightarrow (\north_{i(\north)}\\ & \wedge \south_{i(\south)}\wedge \east_{i(\east)}\wedge \west_{i(\west)} ))\\
\mathit{tile\_match} := {} & \square_b \bigwedge_{c\in C}((\north_c\rightarrow \square_\up \south_c)
%\wedge (\down_c\rightarrow \square_\down \up_c)\\
\wedge (\west_c\rightarrow \square_\leftt \east_c) 
%\wedge (\rightt_c\rightarrow \square_\rightt \leftt_c)
)
\end{align*}}
\end{definition}

The first line of $\psi_\mathit{types}$ guarantees that each state other than the origin has exactly one successor for each direction $x\in D$, and that $b$ is transitive over all $x\in D$ (by which we mean: if $(s_0,s_1)\in R(b)$ and $(s_1,s_2)\in R(x)$, then $(s_0,s_2)\in R(b)$). The second line of $\psi_\mathit{types}$ then guarantees that the four directions form a grid. Finally, the third line of $\psi_\mathit{types}$ guarantees that the grid is tiled in an appropriate way.

%We will show that $\psi_\mathit{tiles}$ is satisfiable if and only if $\mathit{tiles}$ can tile $\mathbb{Z}\times\mathbb{Z}$. The first two lines of $\psi_\mathit{tiles}$ force any model that satisfies $\psi_\mathit{tiles}$ to have a structure similar to $\mathbb{Z}\times\mathbb{Z}$, the last four conjuncts can then be satisfied on such a structure if and only if a tiling exists.

We need to show that $\psi_\mathit{types}$ is satisfiable if and only if $\mathit{types}$ can tile $\mathbb{Z}\times \mathbb{Z}$. We start by showing that if such a tiling exists, then $\psi_\mathit{types}$ is satisfiable.

%We start by showing that if a tiling exists, then $\psi_\mathit{tiles}$ is satisfiable.

\begin{lemma}
\label{lemma:one_way}
Suppose $\mathit{types}$ can tile $\mathbb{Z}\times\mathbb{Z}$. Then $\psi_\mathit{types}$ is satisfiable.
\end{lemma}
\begin{proof}
Let $\mathit{tiling}$ be the tiling, let $p_{n,m}\in \mathcal{P}$ for every $n,m\in \mathbb{Z}$ and let $\M=(S,R,V)$ be the following, quite straightforward, encoding of $\mathit{tiling}$:
\begin{itemize}
	\item $S = (\mathbb{Z}\times\mathbb{Z})\cup s_0$
	\item $R(a) = \{(s,s)\mid s\in S\}$
	\item $R(b) = \{(s_0,(n,m))\mid n,m\in \mathbb{Z}\}\cup \{((n,m),s_0)\mid n,m\in \mathbb{Z}\}$
	\item $R(\up) = \{((n,m),(n,m+1))\mid n,m\in \mathbb{Z}\}$
	\item $R(\down) = \{((n,m),(n,m-1))\mid n,m\in \mathbb{Z}\}$
	\item $R(\leftt) = \{((n,m),(n-1,m))\mid n,m\in \mathbb{Z}\}$
	\item $R(\rightt) = \{((n,m),(n+1,m))\mid n,m\in \mathbb{Z}\}$
	\item $V(p)= \{s_0\}$
	\item $V(p_i) = \{(n,m)\mid \mathit{tiling}(n,m)=i\}$ for $i\in \mathit{tiles}$
	\item $V(\north_c) = \{(n,m)\mid (\mathit(tiling)(n,m))(\north)=c\}$ for $c\in C$
	\item $V(\south_c) = \{(n,m)\mid (\mathit(tiling)(n,m))(\south)=c\}$ for $c\in C$
	\item $V(\east_c) = \{(n,m)\mid (\mathit(tiling)(n,m))(\east)=c\}$ for $c\in C$
	\item $V(\west_c) = \{(n,m)\mid (\mathit(tiling)(n,m))(\west)=c\}$ for $c\in C$
	\item $V(p_{n,m})=\{(n,m)\}$
\end{itemize}

As mentioned above, the state $s_0$ is special: it is the origin state, and the only state that does not have a tile type associated with it. The propositional variable $p$ is used to identify this special state. We will show that $\M,s_0\models \psi_\mathit{types}$.

Note that $R(a)$ is the identity relation.
%Note that every $s\in S$ has a reflexive $a$-arrow. 
As a result, there can be no arrow update that retains the $a$-arrow from $s$ to some state $s'$ but removes the $a$-arrow from $s'$. So $\M\models \lozenge_a \lozenge_a\top \wedge [\AAUL]\neg \lozenge_a\square_a\bot$ and therefore, by definition, $\M\models \refa$. 

Now, consider the pointed model $\M,s_0$. It is straightforward to verify that $\M,s_0\models p \wedge \lozenge_b \top \wedge \square_b \neg p$.
Because we already determined that $\M,s_0\models \refa$, this suffices to show that $\M,s_0\models \psi_1$.

Additionally, $\M\models\refa$ together with the fact that there are $b$-arrows from $s_0$ to every $(n,m)$ and from every $(n,m)$ to $s_0$, which satisfies $p$, implies that $\M,s_0\models \square_b(\refa\wedge \lozenge_bp)$. Furthermore, every state $(n,m)$ has only one outgoing $b$-arrow, namely the one to $s_0$. So for any arrow update $U$, if $U$ retains the $a$-arrow on $s_0$ then it must also retain the $a$-arrow on any state that is $b$-reachable from a state $(n,m)$, since it is the same $a$-arrow. So $\M,s_0\models [\AAUL](\lozenge_a\top \rightarrow \square_b\square_b\lozenge_a\top)$. We have shown that $\M,s_0$ satisfies both conjuncts of $\psi_2$, so $\M,s_0\models \psi_2$.

Now, let us look at $\psi_{3,x}$. For every direction $x$ and every $(n,m)$, there is an $x$-arrow from $(n,m)$ to some $(n',m')$. Like $(n,m)$, this $(n',m')$ satisfies $\neg p \wedge \refa\wedge \lozenge_bp$, so we have $\M,(n,m)\models \lozenge_x (\neg p \wedge \refa\wedge \lozenge_bp)$. For every direction there is only one such successor $(n',m')$. This implies that every arrow update either removes the $a$-arrow from all $x$-successors of $(n,m)$, or it retains the $a$-arrow of all such successors. This implies that $\M,(n,m)\models [\AAUL\nolinebreak](\lozenge_x\lozenge_a\top\rightarrow \square_x\lozenge_a\top)$. This is true for any $(n,m)$, so we have $\M,s_0\models \psi_{3,x}$ for every direction $x\in D$.

All states $(n,m)$ have a $b$-arrow to the same state $s_0$. Once again we use the fact that every arrow update either eliminates or retains the $a$-arrow on this single world, it cannot do both. It follows that $\M,s_0\models [\AAUL](\lozenge_a\top \rightarrow \square_b\square_x\square_b\lozenge_a\top)$, so, by definition, $\M,s_0\models \psi_{4,x}$.

Now we get to the more difficult part, where we need to show that $s_0$ satisfies the rather complicated formula $\propd_x$, for every direction $x$. For ease of notation, we will show that $\propd_\rightt$ holds, the proof for the other directions is similar. Fix any $(n,m)$, and let $U$ be any update such that the antecedent of the implication in $\propd_\rightt$ holds after $U$, i.e.
\begin{align*}\M*U,(n,m)\models {} & \square_a\bot \wedge \lozenge_\rightt \lozenge_a\top \wedge \lozenge_b(\lozenge_b\top\wedge \square_b\lozenge_a\top)\wedge\\ & \langle \AAUL\rangle(\lozenge_\rightt \lozenge_a\top\wedge\lozenge_b\lozenge_b\square_a\bot)\end{align*}
The assumption that this antecedent holds places some restrictions on $U$. The first conjunct says that the $a$-arrow on $(n,m)$ is removed. The second conjunct says that the $\rightt$-arrow from $(n,m)$ to $(n+1,m)$ is retained, and so is the $a$-arrow on $(n+1,m)$. The third conjunct says that the $b$-arrow from $(n,m)$ to $s_0$ is retained, that at least one $b$-arrow from $s_0$ to some $(n',m')$ is retained and that every $(u,v)$ that is $b$-accessible from $s_0$ still has its $a$-arrow.
The final conjunct then says that there is some arrow update, call it $U'$, that
\begin{itemize}
	\item retains the $\rightt$-arrow from $(n,m)$ to $(n+1,m)$ and the $a$-arrow on $(n+1,m)$,
	\item retains the $b$-arrow from $(n,m)$ to $s_0$ and the $b$-arrow from $s_0$ to some state $(u,v)$ and
	\item removes the $a$-arrow on $(u,v)$.
\end{itemize} %that retains the $a$-arrow from $(n+1,m)$ while removing the $a$-arrow from $(n',m')$.
In particular, the fact that the $a$-arrow on $(u,v)$ is removed while that on $(n+1,m)$ is retained shows that, $(u,v)\not = (n+1,m)$.

For any such $U$, we need to show that the consequent is true in $(n,m)$, i.e. that 
\begin{align*}\M*U,(n,m)\models [U_\rightt]\langle \AAUL\rangle(\lozenge_\rightt \lozenge_a\top\wedge\lozenge_b\lozenge_b\square_a\bot).\end{align*}
Recall that $U_\rightt = (p\vee \square_a\bot,b,\top),(\top,a,\top),(\square_a\bot,\rightt,\top)$. 
%So we need to show that after $U_x$ has been applied, there is still some arrow update $U''$ with the same properties as $U'$ discussed above.
So $U_\rightt$ retains all $b$-arrows from $s_0$ (since $\M*U,s_0\models p$), the $b$-arrow from $(n,m)$ to $s_0$ and the $\rightt$-arrow from $(n,m)$ to $(n+1,m)$ (since $\M*U,(n,m)\models \square_a\bot$) and every $a$-arrow---and therefore in particular the $a$-arrows on $(n,m)$ and $(u,v)$. Because $(u,v)\not = (n+1,m)$, the atom $p_{(n+1,m)}$ holds in $(n+1,m)$ but not in $(u,v)$. The update $U'' := (\top,\rightt,\top),(\top,b,\top),(p_{(n+1,m)},a,\top)$ therefore retains all $\rightt$ and $b$ arrows as well as the $a$-arrow on $(n+1,m)$ while removing the $a$-arrow on $(u,v)$. Since $\M*U$, and therefore also $(\M*U)*U''$, contains a $b$-arrow from $(n,m)$ to $s_0$ and from $s_0$ to $(u,v)$ as well as a $\rightt$-arrow from $(n,m)$ to $(n+1,m)$ we have $(\M*U)*U_\rightt,(n,m)\models [U''](\lozenge_\rightt \lozenge_a\top\wedge\lozenge_b\lozenge_b\square_a\bot)$ and therefore
$\M*U,(n,m)\models [U_\rightt]\langle \AAUL\rangle(\lozenge_\rightt \lozenge_a\top\wedge\lozenge_b\lozenge_b\square_a\bot)$,
which was to be shown.

This completes the proof that $\M,s_0\models \propd_\rightt$. 
%In order for $\propd_\east$ to hold in $s_0$, it must be the case that for every $U$, the consequent of $\propd_\east$ then holds, so
%\begin{align*}\M*U,(n,m)\models [U_\east]\langle \AAUL\rangle(\lozenge_\east \lozenge_a\top\wedge\lozenge_b\lozenge_b\square_a\bot).\end{align*}
%Recall that $U_\east = (p\vee \square_a\bot,b,\top),(\top,a,\top),(\square_a\bot,\east,\top)$. So, when applied in $\M*U$, the update $U_\east$ retains the $\east$-arrow from $(n,m)$ to $(n+1,m)$, all $a$-arrows (that were still present in $\M*U$), all $b$-arrows from $s_0$ (because $\M*U,s_0\models p$) and the $b$-arrow from $(n,m)$ (because $\M*U,(n,m)\models \square_a\bot$). All other arrows are removed.
%
%%So $U_\mathit{rem}$ removes all arrows other than the $a$-arrow from $(n+1,m)$ and $(n',m')$ while retaining the $\rightt$ arrow from $(n,m)$ to $(n+1,m)$ as well as the $b$-arrows from $(n,m)$ to $s_0$ and from $s_0$ to $(n',m')$. 
%Afterwards, there must still be some update that retains the $a$-arrow on $(n+1,m)$ while removing it from $(n',m')$. This update does indeed exist because each state is identified by a unique propositional variable. The update $(\top,b,\top),(\top,\rightt,\top),(p_{n+1,m},a,\top)$ suffices, for example. So $M,s_0\models \propd_\east$.
We continue by showing that $\M,s_0\models \return_x$. Again, we show that $\M,s_0\models \return_\rightt$, the other directions can be proven similarly. The formula $\return_\rightt$ holds in $\M,s_0$ if
\begin{align*}\M,(n,m)\models {} & \langle \AAUL\rangle (\square_a\bot \wedge \lozenge_b\top \wedge \lozenge_\rightt \langle\AAUL\rangle(\lozenge_a\top \wedge \lozenge_b (\lozenge_b\top\wedge \square_b\lozenge_a\top)\wedge \\ & [\AAUL](\lozenge_a\top\rightarrow \square_b\square_b\lozenge_a\top)))\end{align*}
for every $(n,m)$. Let $U_1 := (\top,\rightt,\top), (\top,b,p\vee p_{n+1,m}),(p_{(n+1,m)},a,\top)$. The only $b$-arrows that are retained by $U_1$ are those that go to $s_0$ or to $(n+1,m)$. Because the only remaining $b$-arrow from $s_0$ is to $(n+1,m)$, every arrow update that retains the $a$-arrow on $(n+1,m)$ must do so in every $b$-successor of every $b$-successor of $(n+1,m)$, since $(n+1,m)$ is the only such $b$-$b$-successor. So $\M*U_1,(n+1,m)\models [\AAUL](\lozenge_a\top \rightarrow \square_b\square_b\lozenge_a\top)$. Furthermore, $U_1$ retains the $a$-arrow on $(n+1,m)$ so, again using the fact that $(n+1,m)$ is the only remaining $b$-successor of $s_0$, we also have $\M*U_1,(n+1,m)\models \lozenge_a\top\wedge \lozenge_b(\lozenge_b\top \wedge \square_b\lozenge_a\top)$.

Putting these conjuncts together, and using the fact that $\models \chi \rightarrow \langle \AAUL\rangle \chi$ for every $\chi$, we get \begin{equation*}\M*U_1,(n+1,m)\models \langle \AAUL\rangle (\lozenge_a\top\wedge \lozenge_b(\lozenge_b\top \wedge \square_b\lozenge_a\top)\wedge [\AAUL](\lozenge_a\top \rightarrow \square_b\square_b\lozenge_a\top)).\end{equation*}
The update $U_1$ also eliminates the $a$-arrow on $(n,m)$ while retaining the $b$-arrow from $(n,m)$ to $s_0$ and the $\rightt$-arrow from $(n,m)$ to $(n+1,m)$, so
\begin{align*}\M*U_1,(n,m)\models {} & \square_a\bot \wedge \lozenge_b\top \wedge \lozenge_\rightt \langle\AAUL\rangle(\lozenge_a\top \wedge \lozenge_b (\lozenge_b\top\wedge \square_b\lozenge_a\top)\wedge \\ & [\AAUL](\lozenge_a\top\rightarrow \square_b\square_b\lozenge_a\top)).\end{align*}
We have $U_1\in \mathcal{L}_\mathit{AUL}$, so
\begin{align*}\M,(n,m)\models {} & \langle \AAUL\rangle(\square_a\bot \wedge \lozenge_b\top \wedge \lozenge_\rightt \langle\AAUL\rangle(\lozenge_a\top \wedge \lozenge_b (\lozenge_b\top\wedge \square_b\lozenge_a\top)\wedge \\ & [\AAUL](\lozenge_a\top\rightarrow \square_b\square_b\lozenge_a\top))),\end{align*}
which was to be shown in order to prove that $\M,s_0\models \mathit{return}_\rightt$.

We continue with $\mathit{inverse}$. In $\M$, the relations $\up$ and $\down$ are each others inverses, as are $\leftt$ and $\rightt$. Furthermore, all four direction relations are functions. It follows immediately that, for every $(n,m)$, we have $\M,(n,m)\models [\AAUL](\square_a\bot\rightarrow \square_\rightt\square_\leftt\square_a\bot)$, and similarly for the other combinations of directions. So we have $\M,s_0\models \mathit{inverse}$.

Similarly, in $\M$ we have $R(\rightt)\circ R(\up)= R(\up)\circ R(\rightt)$, and the same for the other directions. It follows that $\M,s_0\models \square_b [\AAUL]\bigwedge_{(x,y)\in E}(\lozenge_x\lozenge_y\square_a\bot \rightarrow \square_y\square_x\square_a\bot)$.

The last four conjuncts of $\psi_\mathit{types}$ simply encode the fact that $\mathit{tiling}$ is a tiling on $\mathbb{Z}\times\mathbb{Z}$, so $\M,s_0$ satisfies those as well.
\end{proof}

Left to show is that if $\psi_\mathit{types}$ is satisfiable, then $\mathit{types}$ can tile the plane. We do this by showing that any model where $\psi_\mathit{types}$ is satisfied looks like the model $\M$ that we constructed above. There are some differences, certainly. For example, we cannot enforce that a state $(n,m)$ has exactly one $x$-successor for every direction $x$, only that it has at least one $x$-successor and that all its $x$-successors are indistinguishable. Such differences are not relevant to the existence of a tiling, however.
\begin{lemma}
\label{lemma:other_way}
If $\psi_\mathit{types}$ is satisfiable, then there is a tiling on $\mathbb{Z}\times\mathbb{Z}$ for $\mathit{types}$.
\end{lemma}
\begin{proof}
Suppose $\psi_\mathit{types}$ is satisfiable. Then there is some pointed model $\M,s_0$ such that $\M,s_0\models \psi_\mathit{types}$.

First, consider any state $s$ such that $\M,s$ satisfies $\refa$. Then there is some $s'$ that is $a$-accessible from $s$. Furthermore, by $[\AAUL]\neg \lozenge_a\square_a\bot$, there is no arrow update that retains the $a$-arrow from $s$ to $s'$ while removing any $a$-arrows from $s'$. It follows that there is no modal formula $\phi$ that holds on $s$ but not on $s'$, since otherwise $(\phi,a,\top)$ would have been such an arrow update. So any state $s$ that satisfies $\refa$ has an $a$-arrow to a state that is indistinguishable from $s$, and it only has $a$-arrows to such indistinguishable states.

In order to make the proof easier to read, it is convenient to draw figures of $\M$ as we use $\psi_\mathit{types}$ to construct it. In these figures, we use the following rules for simplification:
\begin{enumerate}
	\item We do not draw the entire model, but only those parts of it that are of interest for the part of the proof illustrated by the figure. So unless explicitly stated otherwise, the lack of a drawn arrow between two states does not always indicate that there is no arrow between them.
	\item When two states are modally indistinguishable, we draw them as if they were one state. So, strictly speaking, when we draw a state $s$ we mean an equivalence class $[s]$ of states.
	\item Most states that we consider satisfy $\refa$. 
	Together with the second rule of simplification, this means we could draw a reflexive $a$-arrow on those worlds. For simplicity of presentation we do not draw these arrows.
\end{enumerate}
Now, suppose $\M,s_0\models \psi_\mathit{types}$. The first conjunct of $\psi_\mathit{types}$ states that $\M,s_0\models \refa \wedge p \wedge \lozenge_b\top\wedge\square_b\neg p$. This means that, so far, our model looks something like the following.
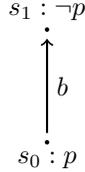
\begin{center}
\begin{tikzpicture}
\fill (0,0) circle (0.03) node (s0) {} node[below] {$s_0: p$} node[right] {};
\fill (0,1.5) circle (0.03) node (s1) {} node[above] {$s_1: \neg p$} node[right] {};

\draw[->,thick] (s0) -- (s1) node[midway, right] {$b$};
\end{tikzpicture}
\end{center}

Next, consider $\psi_2$, the second conjunct of $\psi_\mathit{types}$, which states that $\M,s_0\models \square_b (\refa \wedge \lozenge_b p)\wedge [\AAUL](\lozenge_a\top \rightarrow \square_b\square_b\lozenge_a \top)$. In particular, it implies that any state $s_1$ that is $b$-accessible from $s_0$ satisfies $\refa$, and has a $b$-accessible $p$ state. 

Suppose towards a contradiction that there is any $b$-successor $s'$ of $s_1$ that is modally distinguishable from $s_0$. Then there is some $\phi$ that holds on $s_0$ but not on $s'$. Consider then the update $U = (\phi,a,\top),(\top,b,\top)$. We have $M,s_0\models [U](\lozenge_a\top \wedge \lozenge_b\lozenge_b\square_a\bot)$, contradicting $M,s_0\models [\AAUL](\lozenge_a\top \rightarrow \square_b\square_b\lozenge_a \top)$. Our assumption that such distinguishable $s'$ exists was therefore false, so our model looks like this:
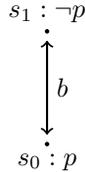
\begin{center}
\begin{tikzpicture}
\fill (0,0) circle (0.03) node (s0) {} node[below] {$s_0: p$} node[right] {};
\fill (0,1.5) circle (0.03) node (s1) {} node[above] {$s_1: \neg p$} node[right] {};

\draw[<->,thick] (s0) -- (s1) node[midway, right] {$b$};
\end{tikzpicture}
\end{center}

Now, let us consider $\psi_{3,x}$, which tells us that $\M,s_0\models \square_b (\lozenge_x (\neg p \wedge \refa \wedge \lozenge_b p) \wedge [\AAUL](\lozenge_x\lozenge_a\top\rightarrow \square_x\lozenge_a\top))$ for every direction $x\in D$. Firstly, from $\M,s_1\models \lozenge_x (\neg p \wedge \refa \wedge \lozenge_b p)$ we learn that there is a state $s_2$ that is accessible from $s_1$, and that this $s_2$ satisfies $\neg p$, $\refa$ and $\lozenge_b p$. Furthermore, by $\M,s_0\models [\AAUL](\lozenge_x\lozenge_a\top\rightarrow \square_x\lozenge_a\top)$ we know that this $s_2$ is unique (up to modal indistinguishability), since otherwise there would be some arrow update that retains the $a$-arrow on $s_2$ but removes it from some other $x$-successor of $s_1$.

Let us then look at the conjunct $\psi_{4,x}$, which states that $\M,s_0\models [\AAUL](\lozenge_a \top \rightarrow \square_b\square_x\square_b\lozenge_a\top)$. This tells us that every $b$-successor of $s_2$ must be indistinguishable from $s_0$, as otherwise there would be an arrow update removing the $a$-arrow from such a successor while retaining the $a$-arrow on $s_0$. Our drawing of the model therefore becomes as follows.
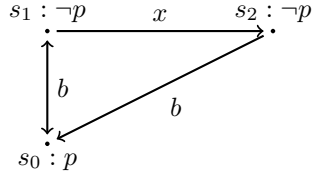
\begin{center}
\begin{tikzpicture}
\fill (0,0) circle (0.03) node (s0) {} node[below] {$s_0: p$} node[right] {};
\fill (0,1.5) circle (0.03) node (s1) {} node[above] {$s_1: \neg p$} node[right] {};
\fill (3,1.5) circle (0.03) node (s2) {} node[above] {$s_2: \neg p$} node[right] {};

\draw[<->,thick] (s0) -- (s1) node[midway, right] {$b$};
\draw[->,thick] (s1) -- (s2) node[midway, above] {$x$};
\draw[->,thick] (s2) -- (s0) node[midway, below right] {$b$};
\end{tikzpicture}
\end{center}

Now, we get to the hard part. The formula $\propd_x$, which is the fifth conjunct of $\psi_\mathit{types}$, tells us that
\begin{align*}\M,s_0\models & \square_b [\AAUL] ((\square_a\bot \wedge \lozenge_x \lozenge_a\top \wedge \lozenge_b(\lozenge_b\top\wedge \square_b\lozenge_a\top)\wedge \\
&\langle \AAUL\rangle(\lozenge_x \lozenge_a\top\wedge\lozenge_b\lozenge_b\square_a\bot))\rightarrow [U_\mathit{x}]\langle \AAUL\rangle(\lozenge_x \lozenge_a\top\wedge\lozenge_b\lozenge_b\square_a\bot))\end{align*}
where
\begin{align*}U_\mathit{x} := {} & (p\vee \square_a\bot,b,\top),(\top,a,\top),(\square_a\bot,x,\top).\end{align*}
The formula starts with $\square_b$, which takes us to any $b$-successor of $s_0$. Let us assume without loss of generality that this is the $s_1$ we drew earlier. In $s_1$, it must be the case that any update $U_1$ that makes $(\square_a\bot \wedge \lozenge_x \lozenge_a\top \wedge \lozenge_b(\lozenge_b\top\wedge \square_b\lozenge_a\top)\wedge \langle \AAUL\rangle(\lozenge_x \lozenge_a\top\wedge\lozenge_b\lozenge_b\square_a\bot)$ true also makes $[U_\mathit{x}]\langle \AAUL\rangle(\lozenge_x \lozenge_a\top\wedge\lozenge_b\lozenge_b\square_a\bot)$ true. Let us look at what $M*U_1$ would look like. Firstly, we have $M*U_1,s_1\models \square_a\bot \wedge \lozenge_x\lozenge_a\top$. So the $a$-arrow on $s_1$ is removed, while the $x$-arrow to $s_2$ and the $a$-arrow on $s_2$ are retained.

Furthermore, $\M*U_1,s_1\models \lozenge_b(\lozenge_b\top\wedge \square_b\lozenge_a\top)$, so the $b$-arrow from $s_1$ to $s_0$ is retained, at least one $b$-arrow from $s_0$ is retained and every $b$-arrow from $s_0$ that still exists points to a $\lozenge_a$ state. So far, the situation in $\M*U_1$ can be drawn as in the following diagram, where $s_3$ is some state $b$-accessible from $s_0$.

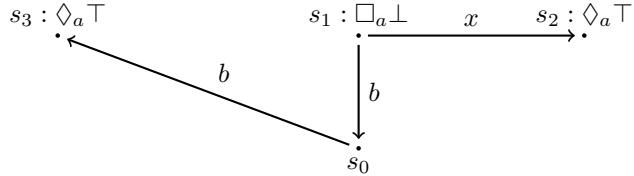
\begin{center}
\begin{tikzpicture}
\fill (0,0) circle (0.03) node (s0) {} node[below] {$s_0$} node[right] {};
\fill (0,1.5) circle (0.03) node (s1) {} node[above] {$s_1: \square_a\bot$} node[right] {};
\fill (3,1.5) circle (0.03) node (s2) {} node[above] {$s_2: \lozenge_a\top$} node[right] {};
\fill (-4,1.5) circle (0.03) node (s3) {} node[above] {$s_3: \lozenge_a\top$} node[right] {};

\draw[<-,thick] (s0) -- (s1) node[midway, right] {$b$};
\draw[->,thick] (s0) -- (s3) node[midway, above right] {$b$};
\draw[->,thick] (s1) -- (s2) node[midway, above] {$x$};
%\draw[->,thick] (s2) -- (s0) node[midway, below right] {$b$};
\end{tikzpicture}
\end{center}

Finally, $\M*U_1,s_1\models \langle \AAUL\rangle(\lozenge_x \lozenge_a\top\wedge\lozenge_b\lozenge_b\square_a\bot)$. So there is some arrow update $U_2$ that retains the $a$-arrow on $s_2$ while removing the $a$-arrow on $s_3$. Such $U_2$ exists if and only if $s_2$ and $s_3$ are modally distinguishable in $\M*U_1$. The formula $\propd_x$ states that for every such $U_1$, we must have $\M*U_1,s_1\models [U_\mathit{x}]\langle \AAUL\nolinebreak\rangle(\lozenge_x \lozenge_a\top\wedge\lozenge_b\lozenge_b\square_a\bot)$. The update $U_\mathit{x}$ is designed in such a way that it removes all arrows, other than the reflexive $a$-arrow, from $s_2$ and $s_3$ while retaining the $x$ arrow from $s_1$ to $s_2$ as well as the $b$-arrows from $s_1$ to $s_0$ and from $s_0$ to $s_3$. So, using the simplification rules, the drawing of $(\M*U_1)*U_\mathit{x}$ is actually the same as the drawing of $\M*U_1$ given above.

However, there is a difference. In $\M*U_1$, we did not draw any arrows from $s_2$ and $s_3$ because there are not guaranteed to be such arrows. In $(\M*U_1)*U_\mathit{x}$ we did not draw any such arrows because they are guaranteed not to exist.

We have $(\M*U_1)*U_\mathit{x},s_1\models\langle \AAUL\rangle\lozenge_x \lozenge_a\top\wedge\lozenge_b\lozenge_b\square_a\bot$, which is the case if and only if $s_2$ and $s_3$ are distinguishable in $(\M*U_1)*U_\mathit{x}$. Since, in $(M*U_1)*U_\mathit{x}$, $s_2$ and $s_3$ have no outgoing arrows other than the $a$-arrow to a state that is indistinguishable from them, it follows that $s_2$ and $s_3$ must be propositionally distinguishable.

In summary, $\propd_x$ states that, for any arrow update $U_1$, if $M*U_1$ matches the drawing shown above and $s_2$ is distinguishable (in $\M*U_1$) from some $b$-successor (in $\M*U_1$) $s_3$ of $s_0$, then it is propositionally distinguishable from such a successor.

Furthermore, we can show that the states also have to be propositionally distinguishable if they are distinguishable in $\M$ (as opposed to $\M*U_1$). Suppose that $U_1$ is such that $\M*U_1$ matches the drawing and that $s_2$ is distinguishable from $s_3$ in $\M$. Now, let $\chi$ be any modal formula that distinguishes between $s_2$ and $s_3$ and (by negating if necessary) assume that $\chi$ holds on $s_1$. We distinguish between three cases:
\begin{itemize}
	\item Suppose $U_1$ removes the $b$-arrows from both $s_2$ and $s_3$. Then let $U_1':= U_1\cup \{(\chi\wedge \neg p,a,\top)\}$.
	\item Suppose $U_1$ removes the $b$-arrow from one of $s_2$ and $s_3$, while retaining the other. Then let $U_1':=U_1$.
	\item Suppose $U_1$ retains the $b$-arrows on both $s_2$ and $s_3$. Then let $U_1'$ be the update obtained by replacing every clause $(\phi,b,\psi)\in U_1$ by $(\phi\wedge (p\vee \chi),b,\psi)$.
\end{itemize}
In any of the three cases, $\M*U_1'$ matches the figure and, furthermore, exactly one of $s_2$ and $s_3$ has a $b$-successor in $\M*U_1'$. In particular, $s_2$ and $s_3$ are distinguishable in $\M*U_1'$. It follows that $s_2$ and $s_3$ are propositionally distinguishable.
%Then we create an update $U_1'$, based on $U_1$, such that $\M*U_1'$ matches the drawing and, furthermore, exactly one of $s_2$ and $s_3$ retains its $b$-arrow to $s_0$. (Such $U_1'$ can be obtained by adding replacing every clause $(\phi,b,\psi)$ of $U_1$ with $(\phi,b,\psi\wedge (p\vee \chi))$ in $U_1$, where $\chi$ distinguishes between $s_2$ and $s_3$.)
%Then we can modify $U_1$ to an arrow update $U_1'$ that removes the $b$-arrow from $s_3$ to $s_0$ while retaining the $b$-arrow from $s_2$ to $s_0$. Then $\M*U_1'$ also matches the drawing, and $s_2$ and $s_3$ are distinguishable in $\M*U_1'$. So $s_2$ and $s_3$ are propositionally distinguishable.

Summarizing once again, using this new observation: for any arrow update $U_1$, if $\M*U_1$ matches the drawing and $s_2$ is distinguishable (in $\M$) from some $b$-successor (in $\M*U_1$) $s_3$ of $s_0$, then $s_2$ is propositionally distinguishable from such a $b$-successor.

Next, we look at $\return_x$. We have
\begin{align*}\M,s_0\models {} & \square_b \langle \AAUL\rangle (\square_a\bot \wedge \lozenge_b\top \wedge \lozenge_x (\lozenge_a\top \wedge \lozenge_b (\lozenge_b\top\wedge \square_b\lozenge_a\top)\wedge \\ & [\AAUL](\lozenge_a\top\rightarrow \square_b\square_b\lozenge_a\top))).\end{align*}
Again, this  starts with a $\square_b$ operator, so we go to any $b$-successor of $s_0$ and assume without loss of generality that this successor is $s_1$. Then there is some update $U_1$ such that $\M*U_1,s_1\models (\square_a\bot \wedge \lozenge_b \top \wedge \lozenge_x (\lozenge_a\top \wedge \lozenge_b (\lozenge_b\top\wedge \square_b\lozenge_a\top)\wedge [\AAUL\nolinebreak](\lozenge_a\top\rightarrow \square_b\square_b\lozenge_a\top)))$.

We will discuss the $[\AAUL](\lozenge_a\top\rightarrow \square_b\square_b\lozenge_a\top)$ part of this formula later. For now, note that the other parts of the formula state that the model $\M*U_1$ can be drawn as follows.
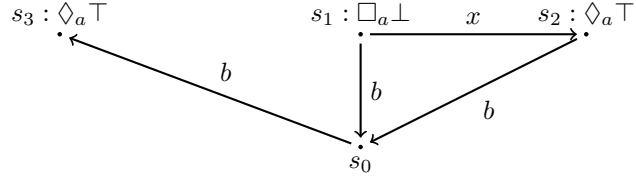
\begin{center}
\begin{tikzpicture}
\fill (0,0) circle (0.03) node (s0) {} node[below] {$s_0$} node[right] {};
\fill (0,1.5) circle (0.03) node (s1) {} node[above] {$s_1: \square_a\bot$} node[right] {};
\fill (3,1.5) circle (0.03) node (s2) {} node[above] {$s_2: \lozenge_a\top$} node[right] {};
\fill (-4,1.5) circle (0.03) node (s3) {} node[above] {$s_3: \lozenge_a\top$} node[right] {};

\draw[<-,thick] (s0) -- (s1) node[midway, right] {$b$};
\draw[->,thick] (s0) -- (s3) node[midway, above right] {$b$};
\draw[->,thick] (s1) -- (s2) node[midway, above] {$x$};
\draw[->,thick] (s2) -- (s0) node[midway, below right] {$b$};
%\draw[->,thick] (s2) -- (s0) node[midway, below right] {$b$};
\end{tikzpicture}
\end{center}

Now recall the diagram drawn further above, when discussing $\propd_x$. The only difference between the two diagrams is that the newer one has an extra $b$-arrow from $s_1$ to $s_0$. By the conventions with which we draw these diagrams, this means the newer diagram is an instance of the older one. So by the fact that $\M,s_0\models \propd_x$, we know that if $s_2$ and $s_3$ are distinguishable in $\M$, then they are propositionally distinguishable so, in particular, that they are distinguishable in $\M*U_1$.

Now, we return to the subformula $[\AAUL](\lozenge_a\top\rightarrow\square_b\square_b\lozenge_a\top)$, which has to hold in $s_2$. It implies that we cannot retain the $a$-arrow on $s_2$ while removing the one on $s_3$. So $s_2$ and $s_3$ are indistinguishable in $\M*U_1$. We already determined that this implies that $s_2$ and $s_3$ are indistinguishable in $\M$ as well.

In particular, this implies that there is some state $s_3$ that is $b$-reachable from $s_0$ and that is indistinguishable from $s_2$ (in $\M$). So $b$ is what could be called \emph{indistinguishable-transitive} over $x$, i.e. for every $(s_0,s_1)\in R(b)$ and every $(s_1,s_2)\in R(x)$ there is some $s_3$ that is modally indistinguishable from $s_2$ such that $(s_0,s_2)\in R(b)$. By the convention to draw indistinguishable worlds as being the same world, this means our diagram of $\M$ simplifies to the following:
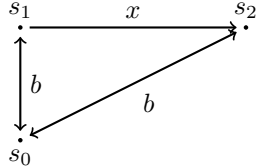
\begin{center}
\begin{tikzpicture}
\fill (0,0) circle (0.03) node (s0) {} node[below] {$s_0$} node[right] {};
\fill (0,1.5) circle (0.03) node (s1) {} node[above] {$s_1$} node[right] {};
\fill (3,1.5) circle (0.03) node (s2) {} node[above] {$s_2$} node[right] {};
%\fill (-4,1.5) circle (0.03) node (s3) {} node[above] {$s_3: \lozenge_a\top$} node[right] {};

\draw[<->,thick] (s0) -- (s1) node[midway, right] {$b$};
\draw[<->,thick] (s0) -- (s2) node[midway, below right] {$b$};
\draw[->,thick] (s1) -- (s2) node[midway, above] {$x$};
%\draw[->,thick] (s2) -- (s0) node[midway, below right] {$b$};
%\draw[->,thick] (s2) -- (s0) node[midway, below right] {$b$};
\end{tikzpicture}
\end{center}

Furthermore, this holds for any $b$-successor $s_1$ of $s_0$. So there is also a unique (up to indistinguishability) $x$-successor $s_3$ of $s_2$ that is $b$-reachable from $s_1$, a unique $x$-successor $s_4$ of $s_3$ that is $b$-reachable from $s_0$ and so on.
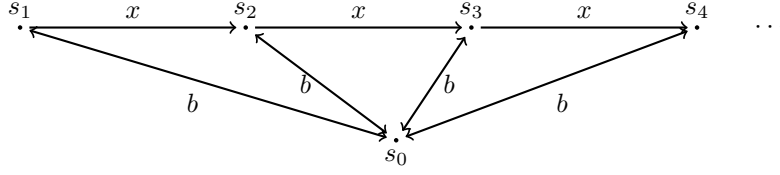
\begin{center}
\begin{tikzpicture}
\fill (0,0) circle (0.03) node (s0) {} node[below] {$s_0$} node[right] {};
\fill (-5,1.5) circle (0.03) node (s1) {} node[above] {$s_1$} node[right] {};
\fill (-2,1.5) circle (0.03) node (s2) {} node[above] {$s_2$} node[right] {};
\fill (1,1.5) circle (0.03) node (s3) {} node[above] {$s_3$} node[right] {};
\fill (4,1.5) circle (0.03) node (s4) {} node[above] {$s_4$} node[right] {};
\node at (5,1.5) {$\cdots$};
%\fill (-4,1.5) circle (0.03) node (s3) {} node[above] {$s_3: \lozenge_a\top$} node[right] {};

\draw[<->,thick] (s0) -- (s1) node[midway, below left] {$b$};
\draw[<->,thick] (s0) -- (s2) node[midway, left] {$b$};
\draw[<->,thick] (s0) -- (s3) node[midway, right] {$b$};
\draw[<->,thick] (s0) -- (s4) node[midway, below right] {$b$};
\draw[->,thick] (s1) -- (s2) node[midway, above] {$x$};
\draw[->,thick] (s2) -- (s3) node[midway, above] {$x$};
\draw[->,thick] (s3) -- (s4) node[midway, above] {$x$};
%\draw[->,thick] (s2) -- (s0) node[midway, below right] {$b$};
%\draw[->,thick] (s2) -- (s0) node[midway, below right] {$b$};
\end{tikzpicture}
\end{center}
Also note that this holds for every direction $x\in D$. So we are approaching the grid-like structure that we need, with $s_0$ being a special point that lies outside the grid.

The next two formulas enforce the grid-like structure, by placing restrictions on certain combinations of directions. Firstly, consider $\mathit{inverse}$. We have
\begin{align*}\M,s_0\models {} & \square_b [\AAUL](\square_a\bot \rightarrow(\square_\up\square_\down\square_a\bot\wedge \square_\down\square_\up\square_a\bot \\ & \wedge \square_\leftt\square_\rightt\square_a\bot \wedge \square_\rightt\square_\leftt\square_a\bot))\end{align*}
Simply put, this formula states that, for every $b$-successor $s_1$ of $s_0$, the unique $\down$-successor of the $\up$-successor of $s_1$ is indistinguishable from $s_1$, and similar for the other combinations of opposite directions. So in $\M$ the opposite directions act as each other's inverses.

Now, consider $\mathit{commute}$. We have
\begin{align*}\M,s_0\models \square_b [\AAUL]\bigwedge_{(x,y)\in E}(\lozenge_x\lozenge_y\square_a\bot \rightarrow \square_y\square_x\square_a\bot)\end{align*}
and $E = \{(\up,\leftt), (\up,\rightt),(\down,\leftt),(\down,\rightt),(\leftt,\up), \allowbreak (\leftt,\down), \allowbreak (\rightt,\up),(\rightt,\down)\}$. This formula simply states that the $\up$-successor of the $\leftt$-successor of $s_1$ is indistinguishable from the $\leftt$-successor of its $\up$-successor, and the same for the other combinations of directions.

Putting all of the above together, we learn that the $b$-successors of $s_0$ form a two-dimensional grid.\footnote{Or, to be slightly more precise: at least one two-dimensional grid. There may be multiple disconnected grids.} The remaining conjuncts of $\psi_\mathit{types}$ straightforwardly encode that this grid is tiled by the set $\mathit{types}$.

The formula $\mathit{one\_tile}$ states that every point in the grid satisfies exactly one propositional variable $p_i$ with $i\in \mathit{types}$. The formula $\mathit{one\_color}$ states that every edge of every point in the grid has at most one color for each direction. The formula $\mathit{tile\_colors}$ guarantees that if a point satisfies $p_i$ then the edge colors of that point are the edge colors of tile $i$. Finally, the formula $\mathit{tile\_match}$ guarantees that the colors of opposing edges of neighboring states match.

All in all, this shows that if $\psi_\mathit{types}$ is satisfiable, then the set $\mathit{types}$ can tile the plane.
\end{proof}

\begin{theorem}
The satisfiability problem of AAUL is co-RE hard.
\end{theorem}
\begin{proof}
Given an instance $\mathit{types}$ of the tiling problem, the formula $\psi_\mathit{types}$ is computable. Furthermore, Lemmas~\ref{lemma:one_way} and \ref{lemma:other_way} show that $\psi_\mathit{types}$ is satisfiable if and only if $\mathit{types}$ can tile the plane. The tiling problem is known to be co-RE complete \cite{berger_1966}, therefore the satisfiability problem of AAUL is co-RE hard.
\end{proof}

\section{Conclusion}
We have shown that the satisfiability of AAUL is uncomputable, like that of similar logics such as APAL \cite{french_2008} and GAL \cite{agotnes_2014}. It is not currently known whether the satisfiability problem of AAUL is co-RE. Typically, one would show that a satisfiability problem is co-RE by providing an axiomatization for the logic, thereby showing the validities of the logic to be RE. However, while there are known axiomatizations for AAUL, APAL and GAL, these axiomatizations are infinitary and therefore cannot be used to enumerate the valid formulas of the logics in question.\footnote{Finitary axiomatizations for APAL and GAL were proposed, in \cite{APAL} and \cite{Agotnes2010} respectively, but these were later shown to be unsound.}

One interesting direction for future research is therefore to determine whether the satisfiability problems of AAUL, APAL and GAL are co-RE complete, and whether these logics admit finitary axiomatizations.

In principle, the proof that we gave for the undecidability of AAUL applies only to the satisfiability problem when considered over the class of all Kripke models. However, we believe that the proof can be adapted to work for the satisfiability problem with respect to other common classes of models such as KD45 and S5.

\bibliographystyle{plain}
\bibliography{arbitraryarrows_v7}

\end{document}